\documentclass[10pt]{article}

\usepackage[english]{babel}
\usepackage[utf8]{inputenc}
\usepackage{graphicx}

\usepackage{amsmath}
\usepackage{amssymb}
\usepackage{upgreek}
\usepackage{bbold} 

\usepackage{amsthm}

\usepackage{authblk} 

\usepackage{enumerate}

\usepackage{xcolor}

\usepackage{algorithm}
\usepackage[noend]{algpseudocode} 

\theoremstyle{definition}

\theoremstyle{plain}
\newtheorem{thm}{Theorem}[section]
\newtheorem{cor}{Corollary}[section]
\newtheorem{lem}{Lemma}[section]
\newtheorem{prop}{Proposition}[section]

\theoremstyle{remark}

\newtheorem{remark}{Remark}[section]

\newcommand{\proba}{\mathbb{P}}
\newcommand{\espe}{\mathbb{E}}
\newcommand{\suce}[2]{\left(#1_#2\right)_{#2\geq1}}
\newcommand{\ndep}{^{(n)}}
\newcommand{\Var}{\mathrm{Var}}
\newcommand{\Cov}{\mathrm{Cov}}

\newcommand{\natu}{\mathbb{N}}
\newcommand{\real}{\mathbb{R}}

\newcommand{\realmas}{\real_{\geq0}}

\begin{document}

\author[2]{Matthieu Jonckheere\thanks{mjonckhe@dm.uba.ar}}
\author[3]{Manuel S\'aenz\thanks{msaenz@ictp.it}}
\affil[2]{\emph{\small{Instituto de C\'alculo, Universidad de Buenos Aires/CONICET, Argentina}}}
\affil[3]{\emph{\small{International Center for Theoretical Physics, UNESCO, Italy}}}

\date{}
\title{\textbf{Exact asymptotic characterisation of running time for \\ approximate gradient descent on random graphs}}
\maketitle

\begin{abstract}
  
  In this work we study the time complexity for the search of local minima in random graphs whose vertices have i.i.d. cost values. We show that, for Erdös-Rényi graphs with connection probability given by $\lambda/n^\alpha$ (with $\lambda > 0$ and $0 < \alpha < 1$), a family of local algorithms that approximate a gradient descent find local minima faster than the full gradient descent. Furthermore, we find a probabilistic representation for the running time of these algorithms leading to asymptotic estimates of the mean running times.
  
  \vspace{0.3cm}
  \noindent
  \it{Keywords: Random Graphs, Optimisation, Local Algorithms, Gradient Descent}
  
\end{abstract}

\section{Introduction}\label{sec:intro}

The last decades saw an exponential increase in the sizes of data sets available on many fields. Many tools and algorithms for managing it have been recently developed. But they usually come with very difficult technical issues, as implementing optimisation routines over spaces of high dimension, which may be extremely computationally demanding. In this context, first-order methods for solving optimisation problems have become a cornerstone of machine learning and operation research.
 
One of the basic (but very common in practice) strategies for minimisation is the celebrated \emph{gradient descent}. It works by sequentially updating the function parameters in the opposite direction of its gradient. This procedure can be shown to lead, under general assumptions, to a local minimum. Because of the computational cost or the impossibility of computing the true gradient with respect to large collections of parameters, variants of the gradient descent are commonly used in practice. One of such variants is the well-known \emph{stochastic gradient descent}, in which the same procedure is replaced by computing a stochastic approximation (usually a sample) of the cost function and which has been successful in machine learning applications \cite{bottou1,bottou2,goodfellow2016deep,sun}.
 
In order to deal with high-dimensionality and complexity, a lot of research has focused on iterative methods that, at each iteration, only optimise over a subset of the decision variables \cite{wright}.
In this framework, several procedures have been proposed: block-coordinate descent methods, 
which alternates the optimisation steps between blocks of selected variables (see also variations around the randomised Kaczmarz algorithm), or randomised block-coordinate
 and alternating randomised block coordinate descent \cite{block,block2}.
 In this last version, the algorithm alternates
 gradient descents steps (or exact optimisation)
 on randomly selected blocks.
Optimisation via sequential exploration is also relevant for
many social, economic, and ecological data sets that are structured as graphs. A crucial task in those contexts is to find vertices that locally minimise/maximise some function over the underlying network. 
For instance, many problems have been modelled as potential games whose Nash equilibria are given by the maxima of a potential \cite{yi2014potential}. These problems have motivated recent studies on optimisation of graph functions. For example, some algorithms were proposed to find extremes of graph functions efficiently, while in \cite{durand2016complexity}, the best response dynamics was studied as an algorithm to find Nash equilibria in random potential games.

\

In this work, we study a toy model based on random graphs to understand possible drastically different regimes concerning the behaviour of this type of strategies. We model the states of the optimisation problem by the nodes of a random network, where the edges correspond to possible
explorations of a new state from a given state.
The fact that randomisation/block coordinates procedures are used is reflected in the randomness of the adjacency matrix of the network.

We make two more modelling assumptions.
First the value of the cost to be optimised will be given by 
a random field (on the nodes of the networks) and the mean number of states explored during the exploration scales according to the connectivity of the networks, which in turn is a function of the total number of nodes. This scaling is such that the
total number of neighbours considered at each exploration step is of order $1$ (which makes sense in practice).

We focus here on the case of i.i.d values for the cost at each node which is of course an over-simplistic assumption. However, it has two main advantages: (i) it allows to perform a very detailed theoretical analysis,
(ii) we believe that the obtained structural results might also be
relevant for more general fields with sufficiently low correlations.
More complex correlations structures will be considered in on-going work.

\paragraph{Contribution}
We study in detail this simple model of sequential exploration on Erdös-Rényi graphs, giving first quantitative results
for approximate gradient optimisation.
We derive the running time complexity for the whole range of parametrisations of the connection probability, finding in each case its asymptotic behaviour. We further show that, in these structures, this algorithm is faster (in a stochastic order sense) than any other \emph{local search algorithm}. In this way, these results provide a possible way of understanding the improvement in the performance achieved by stochastic variants of gradient descent. 
It also allows to give asymptotic estimates of the complexity of finding Nash equilibria in sparse random potential games as considered in \cite{durand2016complexity}.

\

The remaining of the paper is organised as follows. In Section \ref{sec:mainpot} we present the main results of the chapter, how they relate to each other, and we discuss how they characterise the complexity of the algorithm. While in Section \ref{sec:characpot} we present a simplified dynamics that will be of key importance in deriving these results. And finally, in Section \ref{sec:proofspot} we give the proofs for Erdös-Rényi graphs.

\subsection{Description of the process studied}

We focus on Erdös-Rényi graphs with $n$ nodes with connection probabilities $p$ that scale with respect to the graph size according to $p(n)= \lambda n^{-\alpha}$ (with $\lambda > 0$ and $0 < \alpha < 1$). 

We present in Algorithms 1 and 2 respectively the pseudo-code for the Gradient Descent (GD) algorithm and for the local optimisation algorithm studied, which we will refer to as \emph{approximate gradient descent} (AGD).

\begin{algorithm}[ht]
\caption{Gradient Descent (GD)}
\label{algo:1}
\begin{algorithmic}
\State INPUTS: Adjacency matrix $A$ and positive values $U(i)$ for each node $i$
    \State Set $v$ as a uniform vertex in $[n]$
    \State Set $N$ as the second indices of the non-zero elements of $A(v,:)$
    \State Set $min$ as $\mbox{Min}_{i \in N} U(i)$
    \While{$U[v] > min$}
        \State Set $v_{min} = \mbox{ArgMin}_{i \in N} U(i)$
        \State Set $i$ as the number of elements in $N\cup\{v\}$ smaller than $v_{min}$
        \State Remove the vertices $N \cup \{v\} \backslash \{v_{min}\}$ from $A$ and $U$
        \State Set $v = v_{min} - i$
        \State Set $N$ as the second indices of the non-zero elements of $A(v,:)$
        \State Set $min$ as $\mbox{Min}_{i \in N} U(i)$
    \EndWhile
    \Return v
\end{algorithmic}
\end{algorithm}

\begin{algorithm}[ht]
\caption{Approximate Gradient Descent (AGD)}
\begin{algorithmic}
\State INPUTS: Adjacency matrix $A$ and positive values $U(i)$ for each node $i$
    \State Set $v$ as a uniform vertex in $[n]$
    \State Set $k = 0$
    \While{$k < \mbox{Rows}(A)/B$}
        \State Set $k+=1$
      \State Set $I$ as the vector of integers ranging from $(k-1)B$ to $\mbox{Min}(kB,\mbox{Rows}(A))$ 
        \State Set $N$ as the second indices of the non-zero elements of $A(v,I)$
        \State Set $min$ as $\mbox{Min}_{i \in N} U(i)$
        \If{$U(v) > min$}
            \State Set $k = 0$
            \State Set $v_{min}$ equal to $\mbox{ArgMin}_{i \in N} U(i)$
            \State Set $i$ as the number of elements in $N\cup \{v\}$ smaller than $v_{min}$
            \State Remove the vertices $N\cup\{v\}\backslash\{v_{min}\}$ from $A$ and $U$
            \State Set $v = v_{min} - i$
        \Else
            \State Set $i$ as the number of elements in $N$ smaller than $v$
            \State Remove the vertices $N$ from $A$ and $U$
            \State Set $v = v - i$
        \EndIf
    \EndWhile
    \Return v
\end{algorithmic}
\end{algorithm}

While the GD in each step explores the entire neighbourhood of the current vertex, the AGD implements the same dynamic but only exploring in each step {\bf some fixed number $B(n) \in \natu$ of vertices.} For this, at step $i\geq1$, it first determines which of these $B(n)$ vertices explored are connected to the vertex $v_i$ in which the algorithm is currently at. For simplicity, we drop the dependence in $n$ when it is not explicitly needed. The set of these vertices in the neighbourhood of $v_i$ will be called the \emph{partial neighbourhood}  $\mathcal{N}_i$. If any of them has a cost smaller than the one of the current vertex, the algorithm sets $v_{i+1}$ equal to the vertex in $\mathcal{N}_i$ with the smallest cost, and the vertices in $\{v_i\} \cup \{ \mathcal{N}_i \backslash \{v_{i+1}\}\}$ are removed from the graph. Otherwise, if no vertex with lower cost is found in $\mathcal{N}_i$, then the vertices in $\mathcal{N}_i$ are removed from the graph and $v_{i+1}$ is set equal to $v_i$.

The dynamics considered will now be described in terms of a stochastic process over a random graph. As underlined in the introduction, we consider very simple assumptions. However, the resulting toy model already 
displays an interesting behaviour. Given a vertex set $[n]$, we hence assume that the underlying graph is an Erd\"os-R\'enyi graph  $G\ndep\sim\mbox{ER}_n(\lambda/n^\alpha)$. Moreover the values of the cost of each vertex $v \in [n]$ are given by random variables $W_v\sim\mbox{Unif}[0,1]$ independent of each other.

At each step of the dynamics, each vertex of the graph will be said to either be \emph{unexplored} (if it does not belong to a partial neighbourhood already visited by the process) or \emph{explored} (if it does). The dynamics in question will be given by a Markov chain $\left(X\ndep_i,U\ndep_i\right)_{i\geq1}$ defined on the graph $G\ndep$, where $X\ndep_i$ gives the cost of the vertex the algorithm is currently at during step $i$ and $U\ndep_i$ the corresponding number of unexplored vertices. Initially the process is located in a vertex chosen uniformly from $[n]$, hence $X\ndep_1$ is set equal to the cost of that vertex and $U\ndep_1=n-1$. Then, if at step $i>1$ the process is at some vertex $v_i$ with cost $W_{v_i}$, at step $i+1$ we will have that $X\ndep_{i+1} = \min_{v\in\mathcal{N}_i(v_i)\cup \{v_i\}} W_v$, $U\ndep_{i+1}$ will be equal to $U\ndep_i-|\mathcal{N}_i|$.

Suppose that after the $i$-th step, the process is at vertex $v_{i}$. In this case, the number of new cost values explored during the next step will be $E\ndep_i := |\mathcal{N}_i| \sim \mbox{Binom}(B, \lambda/n^\alpha)$. Note that this will not be true for the last partial neighbourhood of a given vertex. Nevertheless, this partial neighbourhoods will appear only a few times, and then the correction given by taking their true sizes will be negligible. Here and in the rest of the paper we choose $B(n)$ to be a function scaling as $\beta\lambda^{-1} n^\alpha$ (for some $\beta > 0$) so that at each step the number of explored vertices is $\Theta(1)$. In this case we will have that $E\ndep_i \sim \mbox{Pois}(\beta) + o_d(1)$.

We define the sequence of jumping times for this dynamics in an inductive way according to $\tau\ndep_0 := 0$ and, for $j\geq1$, $\tau\ndep_{j+1} := \inf\{i \geq 0 : X\ndep_{\tau_j+i} < X\ndep_{\tau_j}\}$. These times give the amount of steps the algorithm remains in each visited vertex. Note that if at step $i\geq1$ the algorithm moves for the $j$-th time to the vertex $v_i$, it will be a local minimum of the cost function iff $\tau\ndep_j = \lceil U\ndep_i/B \rceil$.

We assume that  the following important feature for the dynamics:
\begin{enumerate}
    \item it runs until a vertex that is a local minimum has been reached
    \item
    {\bf and} until the algorithm has explored all its neighbourhood in $G\ndep$, effectively checking it is a minimum.
\end{enumerate}  
As a consequence, the algorithm runs until the stopping time 
\begin{equation}\label{eq:defS}
    \mathcal{S}\ndep := \inf\left\{ i\geq n/B : \exists j\geq1 \ s.t. \ X\ndep_i = X\ndep_{i- \lceil U\ndep_j/B \rceil}\right\}.
\end{equation}
Although our results focus in characterising $\mathcal{S}\ndep$, we also study 
\begin{equation}\label{eq:defV}
    \mathcal{V}\ndep := \sum_{i=1}^{\mathcal{S}\ndep} \mathbb{I}_{\{X\ndep_i > X\ndep_{i-1}\}}
\end{equation}
i.e., the number of vertices visited by the process; and the total number of explored vertices
\begin{equation}
    \mathcal{E}\ndep := \sum_{i=1}^{\mathcal{S}\ndep} E\ndep_i.
\end{equation}

Notice that during each step of the exploration, $B(n)$ new vertices are explored of which approximately $Pois(\beta)$ are connected to the vertex the algorithm is currently at. This means that in each step, although $B(n)$ vertices are explored, only $\beta$ new cost values are revealed in mean. The grouping of the exploration into blocks of size $B(n)$ is just done for convenience during the analysis. As we will see later, the exact value of $\beta$ does not affect the first order asymptotic value of the number $\mathcal{E}\ndep$ of cost values explored during the running of the algorithm.

\paragraph{Coupling of the sequence of processes.} Here and in the rest of the chapter we set $\suce{X\ndep}{i}$ to be a sequence of random processes as the one defined above, each one taking place over a graph $G\ndep$ constructed over the vertex set $[n]$. Let $\suce{W}{j}$ be a sequence of independent uniform $[0,1]$ random variables. We will couple the sequence of processes $\suce{X\ndep}{i}$ in such a way that whenever a new vertex is explored, its cost value is drawn from $\suce{W}{j}$ sequentially. That is, in the whole sequence $\suce{X\ndep}{i}$, the tenth cost value explored will be given by $W_{10}$ independently of the value of $n$. Note that this does not affect the individual distributions of the processes but makes their joint distributions correlated in a specific way. This will later be used to couple the whole sequence to a given instance of a simpler dynamics.

\section{Main results and discussion}\label{sec:mainpot}

Our first result proves that the dynamics in question is faster, in the stochastic order sense, than any other \emph{local search algorithm}. The word \emph{local} here means that at each step the algorithms can only either decide to \emph{stay} at the vertex it currently is at or to \emph{move} to some of the newly discovered vertices in the current neighbourhood it is exploring. Note that by this definition, GD and AGD are local search algorithms.

\begin{lem}\label{lem:BRAoptimality}
    Let $G\ndep,\bar{G}\ndep \sim \mbox{ER}_n(\lambda n^{-\alpha})$ be two independent random graphs, $\mathcal{E}\ndep$ be the number of vertices explored by the AGD before finding a local minimum in $G\ndep$, and $\mathcal{\bar E}$ the analogous quantity for another local search algorithm run on $\bar{G}\ndep$. Then, $\mathcal{E} \leq_{st} \mathcal{\bar E} + \mathcal{O}_\mathbb{P}(1)$.
\end{lem}

This result is analogous to the one established in \cite{durand2016complexity} for  equilibrium finding algorithms in potential games. This lemma, though very simple, is important since it implies that our results characterise the complexity of local searches over Erd\"os-Rényi graphs with uniform costs.

We can now focus on characterising the running times for the AGD algorithm. Our main contribution is the characterisation of the asymptotic running times for the search algorithm. As stated above, because of the optimality of the dynamics studied on the set of local search algorithms, these results determine the complexity of the task. We also give bounds for the variance of the stopping times showing that it does not converge to a constant in the limit. It is interesting to note that, although the settings are very different, the asymptotics of the mean running time is exactly the same as the one found in \cite{durand2016complexity} for the search of equilibria in potential games. 

The following Theorem is instrumental to study the total number of steps $\mathcal{S}^{(n)}$.
It allows to essentially replace it by a functional of independent inhomogeneous Poisson processes. 
Define $K(n) := \lfloor \lambda n^{1-\alpha} \rfloor$. Let (for $i< K(n)$) $\lambda_i(t) := \left( 1 - e^{-t} \right)^{i-1} e^{-t}$ and $\lambda_{K}(t) := \left( 1 - e^{-t} \right)^{K(n)}$. For every $i \in [K(n)]$ let $N^{(i)}_t$ be an independent and inhomogeneous Poisson process of intensity $\lambda_i(t)$. For these Poisson processes define the associated random variable $\mathcal{\tilde S}\ndep := K(n) + \sum_{i=1}^{K(n)-1} i N^{(i)}_{T_K}$, where $T_K := \inf\{t\geq0 : N^{(K)}_t > 0 \}$.

\begin{thm}\label{thm:sparsemean}
    If $G\ndep\sim\mbox{ER}_n(\lambda n^{-\alpha})$ with $\lambda >0$ and $0<\alpha<1$, then 
    \begin{equation*}
        \espe\left(\bigg| \frac{\mathcal{S}\ndep}{n^{1-\alpha}} - \frac{\mathcal{\tilde S}\ndep}{\beta n^{1-\alpha}} \bigg| \right) \xrightarrow{n\rightarrow\infty}0 \ \ \ \ \ \mbox{and} \ \ \ \ \ \espe\left(\big| \mathcal{E}\ndep - \mathcal{\tilde S}\ndep \big| \right) \xrightarrow{n\rightarrow\infty}0.
    \end{equation*}
\end{thm}
    
By means of this theorem we are able to compute the leading term of the mean number of steps and explorations of cost values.
\begin{prop}
    If $G\ndep\sim\mbox{ER}_n(\lambda n^{-\alpha})$ with $\lambda >0$ and $0<\alpha<1$, then
    \begin{equation*}
        \espe(\mathcal{S}\ndep) = \frac{\lambda}{\beta} e^\gamma n^{1-\alpha} (1 + o(1)) \ \ \ \ \ \mbox{and} \ \ \ \ \ \espe(\mathcal{E}\ndep) = \lambda e^\gamma n^{1-\alpha} (1 + o(1)),
    \end{equation*}
    where $\gamma$ is the Euler-Mascheroni constant.
\end{prop}

\begin{remark}
This last proposition together with Lemma \ref{lem:BRAoptimality} completely characterise the first order asymptotics of the complexity of local search algorithms in Erd\"os-R\'enyi graphs with uniform costs.
\end{remark}

\begin{remark}
As already underlined, the parameter $\beta$ has no influence in the first order asymptotics of $\mathcal{E}\ndep$ but might be involved in the $o(1)$ term.
\end{remark}

Also, in Proposition \ref{cor:varStilde} in Section \ref{sec:proofspot}, asymptotic upper and lower bounds for the variance of the stopping times are given. These bounds prove that the variance of $\mathcal{\tilde S}\ndep$ scales as the square of the mean; which implies that, in the large graph limit, these stopping times do not converge to a deterministic value. The convergence to a limiting distribution and its characterisation are left for future research.





\subsection{Strategies of the proofs}

The proof of Theorem \ref{thm:sparsemean} is based on a characterisation of a simpler dynamics, which we will call \emph{cost exploration process}, studied in Section \ref{sec:characpot}. As we will see, $\suce{X}{i}$ can be coupled to be given by this process evaluated on a random set of times. Then, the number of steps is derived from computations for the cost exploration process and bounds for the probabilities of certain events that allow to relate both processes.


\section{Characterisation of the cost exploration process}\label{sec:characpot}

In this section we study a simplified dynamics that will later be used in the proofs of our results. It is linked with \emph{record dynamics} as presented in \cite{shorrock1972record,shorrock1973record,ahsanullah2004record}, also called \emph{intersection free dynamics} in \cite{durand2016complexity}. Some of our computations could be derived from results in \cite{ahsanullah2004record}; nevertheless, we chose to be as self-contained as possible and derive our estimates directly. While the Poisson process introduced to study the performance of the algorithm is new, some related characterisations of the process are present in the literature mentioned above (see for example \cite{shorrock1973record}).

Given a sequence $\suce{W}{i}$ of independent and uniform random variables in $[0,1]$, we define its associated \emph{cost exploration process} as a stochastic process $\suce{\tilde X}{i}$ where, for every $i\geq1$, $\tilde X_i = \max \{W_1,\dots,W_i\}$ and $\tilde X_0 = 0$. Associated to this process, we define the sequence of \emph{jumping times} inductively as $\tilde\tau_0 = 0$ and, for $j\geq1$, $\tilde\tau_{j+1} := \inf\{i \geq 0 : \tilde X_{\tilde \tau_j+i} < \tilde X_{\tau_j}\}$. Here we study the stopping time $\mathcal{\tilde S}\ndep := \inf\{i\geq A_n: \tilde X_i = \tilde X_{i-A_n}\}$ and $\mathcal{\tilde V}\ndep := \sup\{j\geq1 : \sum_{k=1}^j \tilde \tau_j \leq \mathcal{\tilde S}\ndep \}$; where $\suce{A}{n}$ is some increasing sequence to be fixed later on. These two quantities will, in the next section, be associated to the corresponding ones of the AGD.

If we set $t_j := \sum_{k=1}^j \tilde \tau_k$ we can define the \emph{jumping process} $\suce{\tilde Y}{j}$ of $\suce{\tilde X}{i}$ by $\tilde Y_j = \tilde X_{t_j}$. The process $\tilde Y_j$ is just the process $\tilde X_i$ restricted to the times in which it changes value (jumps). This last process will be useful because many computations can be performed explicitly for it.

\subsection{Poisson process representation}\label{sec:poischar}

We will now give a characterisation of the jumping process $\suce{\tilde Y}{j}$. It is based on the fact that the jumping process is given by the following recursion
\begin{equation}\label{eq:recjumps}
    \tilde Y_{j+1} = \tilde W_{j+1} \tilde Y_j,
\end{equation}
for some uniform random variable $\tilde W_{j+1}$; i.e., in each step the jump made is going to be uniform in the interval $(0,\tilde Y_j)$. Using this, we can find an explicit characterisation of the jump process in terms of a series of inhomogeneous Poisson processes.

\begin{lem}\label{lem:charjumps}
    Given the process $\suce{\tilde Y}{j}$, we have that $\tilde Y_j = \tilde W_1 \cdots \tilde W_j$, where $\suce{\tilde W}{j}$ are uniform independent random variables in $[0,1]$. In particular, the $-\ln\left(\tilde Y_j\right)$ are distributed according to a Poisson point process of intensity $1$.
\end{lem}

\begin{proof}
A uniform $[0,1]$ random variable $U$, conditioned on being larger than $\tilde Y_j$, is uniform in $[0,\tilde Y_j]$. This justifies the recursion \eqref{eq:recjumps} and the fact that $\tilde Y_j = \tilde W_1 \cdots \tilde W_j$, for some i.i.d. uniform $[0,1]$ random variables $\tilde W_1,\dots,\tilde W_j$.

By taking minus logarithm we get that
\begin{equation*}
    -\ln(\tilde Y_i) = \sum_{j=1}^i -\ln(\tilde W_j).
\end{equation*}
We then reach the conclusion by noting that the terms $-\ln(\tilde W_j)$ are all independent and exponential random variables of parameter $1$. \end{proof}

Let us call $\left(N_t\right)_{t\geq0}$ the Poisson process associated to the jumping process. Note that the $j$-th point $T_j$ in the process is associated to the $j$-th value of the process $\tilde Y_j$ by $T_j = - \ln(\tilde Y_j)$. Also, (for every $j\geq1$) conditionally on $\tilde Y_j$, $\tilde \tau_j$ is distributed as a geometric r.v. of parameter $\tilde Y_j$, that is
\begin{equation}\label{eq:probacontau}
    \proba\left(\tilde\tau_j \geq k | \tilde Y_j\right) = (1-\tilde Y_j)^{k-1}.
\end{equation}

We then have that for each point $T_j$ in the process the associated jumping time is given by a variable $\tilde \tau_j \sim \mbox{Geom}(e^{-T_j})$. These times are, conditionally on the $T_j$, also independent of each other. We can then describe the distribution of the points and the jumping times together by thinnings of the process (see \cite{assunccao2007independence}, for some background on thinnings of Poisson processes).

If we define $N^{(A_n)}_t$ as the point process of points of $N_t$ where the points have a jumping time larger or equal to $A_n$, $N^{(A_n)}_t$ is a thinning of $N_t$ with thinning parameter $\lambda_A(t) = (1-e^{-t})^{A_n-1}$. Likewise, by defining (for $i \in [A_n]$) $N^{(i)}_t$ as the point process of points of $N_t$ that have a jumping time equal to $i$, $N^{(i)}_t$ is a thinning of $N_t$ with thinning parameter $\lambda_i(t) = (1 - e^{-t})^{i-1} e^{-t}$. The processes resulting from these thinnings are independent Poisson processes with intensities given by their respective thinning parameter.

If we call $\bar T_{A_n}$ the first point of the process $N^{(A_n)}_t$, we can then use these processes to express the stopping times we want to estimate as $\mathcal{\tilde V}\ndep = 1 + \sum_{i=1}^{{A_n}-1} N^{(i)}_{\bar T_{A_n}}$ and $\mathcal{\tilde S}\ndep = A_n + \sum_{i=1}^{{A_n}-1} i N^{(i)}_{\bar T_{A_n}}$.


\subsection{Estimation of the jumping times}

Using the characterisation given in Lemma \ref{lem:charjumps} we can derive the distribution of the jumping times $\tilde \tau_j$. Although the result is known (see \cite{ahsanullah2004record}) the proof presented here is new and more in tone with the rest of the paper.

\begin{lem}\label{lem:probatau}
    For $j\geq1$, the jumping time $\tilde \tau_j$ has distribution given by
    \begin{equation*}
        \proba\left(\tilde \tau_j > k \right) = \frac{1}{j!} \int_0^\infty t^{j-1} e^{-t} \left(1-e^{-t}\right)^k dt
    \end{equation*}
\end{lem}
\begin{proof}
By \eqref{eq:probacontau}, we have that
\begin{equation*}
    \proba\left(\tilde\tau_j > k \right) = \espe\left((1-\tilde Y_j)^k\right).
\end{equation*}
Which according to Lemma \ref{lem:charjumps} is given by
\begin{equation*}
    \begin{split}
        \proba\left(\tilde\tau_j > k \right) = \int_{[0,1]^j} \left(1-x_1\cdots x_j\right)^k dv_j & = \int_0^1 \cdots \int_0^1 \left(1-x_1 \cdots x_j\right)^k dx_1 \cdots dx_j \\
        & = \sum_{l=0}^k (-1)^l \binom{k}{l} \int_0^1 x_1^l dx_1 \cdots \int_0^1 x_j^l dx_j \\
        & = \sum_{l=0}^k (-1)^l \binom{k}{l} \frac{1}{(l+1)^j},
    \end{split}
\end{equation*}
where in the second equality we used Newton's binomial. The last line is nothing more than the $k$-th binomial transform (see \cite{boyadzhiev2018binomial} for an introduction on binomial transforms) of the sequence $1/(l+1)^j$. Then, by equality (8.52) in \cite{boyadzhiev2018binomial} we get that
\begin{equation*}
    \proba\left(\tilde\tau_j > k \right) = \frac{1}{j!} \int_0^\infty t^{j-1} e^{-t} \left(1-e^{-t}\right)^k dt,
\end{equation*}
which is what we wanted to prove.
\end{proof}

Making use of this lemma, we now show that the values of the jumping times $\tilde\tau_j$ belong, for every $\delta>0$ and with exponentially high probability, to some deterministic intervals $(m_j(\delta),M_j(\delta))$.

\begin{prop}\label{prop:expboundtau}
    For every $\delta>0$, let $m_j(\delta):= e^{j(1-\delta)}$ and $M_j(\delta):=e^{j(1+\delta)}$. Then, $\proba(\tilde\tau_j \leq m_j(\delta)) = \mathcal{O}(e^{-j})$ and $\proba(\tilde\tau_j \geq M_j(\delta)) = \mathcal{O}(e^{-j})$.
\end{prop}

\begin{proof}
    During the proof we use the easy to check fact that if we have two positive sequences $\suce{a}{j}, \suce{b}{j} \subseteq \realmas$ such that $a_j, b_j \xrightarrow{j\rightarrow\infty} \infty$ and $a_j^2 / b_j \xrightarrow{j\rightarrow\infty} 0$, then
    \begin{equation}\label{eq:cotausual}
        \left( 1 - \frac{a_j}{b_j} \right)^{b_j} = e^{-a_j} ( 1 + \mathcal{O}(a_j^2/b_j) ).
    \end{equation}
    
    For the upper bound, fix $\delta'$ s.t. $0 < \delta' < \delta$, and define $\tilde M_j(\delta') := j e^{j(1+\delta')}$ and $t_j^*(\delta') := j(1+\delta')$. Then, by Proposition \ref{lem:probatau} we have that
    \begin{equation*}
        \begin{split}
            \proba( \tilde\tau_j \geq M_j ) \leq \proba( \tilde\tau_j \geq \tilde M_j ) & = \frac{1}{j!} \int_0^\infty t^{j-1} e^{-t} \left(1-e^{-t}\right)^{\tilde M_j} dt \\
            & \leq \left(1-e^{-t_j^*}\right)^{\tilde M_j} \frac{1}{j!} \int_0^{t^*_j} t^{j-1} e^{-t} dt + \frac{1}{j!} \int_{t^*_j}^\infty t^{j-1} e^{-t} dt,
        \end{split}
    \end{equation*}
    where we used that $(1-e^{-t})^{\tilde M_j}$ is an increasing function bounded by $1$. Now, if we define a random variable $Z_j \sim \Gamma(1,j)$, we can rewrite this as
    \begin{equation*}
        \proba( \tilde\tau_j \geq M_j ) \leq \left(1-e^{-t_j^*}\right)^{\tilde M_j} \proba( Z_j \leq j(1+\delta') ) + \proba( Z_j \geq j(1+\delta') ). 
    \end{equation*}
    But $\Gamma(1,j)$ random variables are distributed as the sum of $j$ independent exponential random variables of parameter $1$. Then, we can apply Cramer's Theorem \cite{dembo2010large} to $Z_j$, which gives
    \begin{equation}\label{eq:tausumable}
        \proba( \tilde\tau_j \geq M_j ) \leq \left(1-\frac{j}{\tilde M_j}\right)^{\tilde M_j} + e^{-I(\delta')j} = e^{-j} (1 + \mathcal{O}(j^2e^{-j(1+\delta')}) + e^{-I(\delta')j}.
    \end{equation}
    Where we used \eqref{eq:cotausual} and $I(\cdot)$ is the rate function of exponential r.v.'s of parameter $1$. The proof of the upper bound of $\proba(\tilde\tau_j \leq m_j(\delta))$ is completely analogous.
\end{proof}

By the bounds provided by Proposition \ref{prop:expboundtau}, the next corollary follows.

\begin{cor}\label{cor:boundV}
    Let $\epsilon>0$. Then, if $A_n \leq n$, we have that $\proba(\mathcal{\tilde V}\ndep > n^{\epsilon}) = \mathcal{O}(e^{- n^\epsilon})$.
\end{cor}
\begin{proof}
    This is the case because for every $\delta > 0$ and $n\geq1$ large enough,
    \begin{equation*}
        \proba(\mathcal{\tilde V}\ndep > n^{\epsilon}) \leq \proba(\tilde \tau_{n^\epsilon} < A_n) \leq \proba(\tilde \tau_{n^\epsilon} < e^{(1-\delta)n^\epsilon}) = \mathcal{O}(e^{- n^\epsilon}).
    \end{equation*}
    Where in the last equality we used the previous proposition.
\end{proof}


\subsection{Estimation of the total running time}

Because of the independence of $\bar T_{A_n}$ from the processes $N^{(i)}_t$ we have the following lemma, which is a variation of Campbell's formula.

\begin{lem}\label{lem:poissonmean}
    Let $\bar T_{A_n}$ and $N^{(i)}_t$ be as before, we then have that for each $i \in [{A_n}-1]$
    \begin{equation*}
        \espe\left(N^{(i)}_{\bar T_{A_n}}\right) = \int_0^\infty \lambda_i(t) e^{-\int_0^t \lambda_A(s)ds} dt.
    \end{equation*}
\end{lem}
\begin{proof}
Define $l_{jk} := [(k-1)/j,k/j)$. Fixing $j\geq1$, then for each $i \in [{A_n}-1]$,
\begin{equation*}
    \sum_{k\geq1} \espe\left( N^{(i)}_{(k-1)/j} \right) \proba(\bar T_{A_n}\in l_{jk}) \leq \espe\left(N^{(i)}_{\bar T_{A_n}}\right) \leq \sum_{k\geq1} \espe\left( N^{(i)}_{k/j} \right) \proba(\bar T_{A_n}\in l_{jk}).
\end{equation*}
Here we used the independence of $N^{(i)}_{k/j}$ from $\bar T_{A_n}$. And because $\espe\left( N^{(i)}_{k/j} \right) = \int_0^{k/j} \lambda_i(t) dt$, by letting $j\rightarrow\infty$ we get that
\begin{equation*}
    \begin{split}
        \espe\left(N^{(i)}_{\bar T_{A_n}}\right) & = \int_0^\infty \left( \int_0^s \lambda_i(t) dt \right) \proba_{\bar T_{A_n}}(ds) \\
        & = \int_0^\infty \lambda_i(t) \proba(\bar T_{A_n} > t ) dt.
    \end{split}
\end{equation*}
Where in the last equality we just changed the order of integration.
\end{proof}

Using this lemma, we can then calculate the mean value of $\mathcal{\tilde S}\ndep$.

\begin{prop}\label{prop:meanconvergence}
   Given the above construction, we have that $\espe\left(\mathcal{\tilde S}\ndep\right) = e^{H_{A_n-1}}$, where (for $n\in\mathbb{N}_0$) $H_n$ is the $n$-th Harmonic Number.
\end{prop}

\begin{proof}
First, note that
\begin{equation*}
    \sum_{i=1}^{{A_n}-1} i \lambda_i(t) = \frac{e^t [1-(1-e^{-t})^{A_n-1}] - (A_n-1) (1-e^{-t})^{A_n-1}}{1-(1-e^{-t})^{A_n-1}}.
\end{equation*}
Then, by Lemma \ref{lem:poissonmean} we have that
\begin{equation*}
    \espe\left(\mathcal{\tilde S}\ndep\right) = I_1 + ({A_n}-1) I_2 + {A_n},
\end{equation*}
where
\begin{equation*}
    I_1 = \int_0^\infty e^t [1-(1-e^{-t})^{A_n-1}] e^{-\int_0^t (1-e^{-s})^{A_n-1}ds} dt = e^{\int_0^\infty 1-(1-e^{-s})^{A_n-1}ds} - 1,
\end{equation*}
and
\begin{equation*}
    I_2 = - \int_0^\infty (1-e^{-t})^{A_n-1} e^{-\int_0^t (1-e^{-s})^{A_n-1}ds} dt = -1.
\end{equation*}
Here the first integral was computed making the change of variables $u = e^{\int_0^t 1-(1-e^{-s})^{A_n-1}ds}$, and the second one by the change $v = e^{-\int_0^t (1-e^{-s})^{A_n-1}ds}$. Then,
\begin{equation*}
    \espe\left(\mathcal{\tilde S}\ndep\right) = e^{\int_0^\infty 1-(1-e^{-s})^{A_n-1}ds}.
\end{equation*}

To prove our result, it suffices to show that $\int_0^\infty 1-(1-e^{-s})^{A_n-1}ds = H_{{A_n}-1}$. By Newton's binomial,
\begin{equation*}
    \begin{split}
        \int_0^\infty 1-(1-e^{-s})^{A_n-1}ds & = \int_0^\infty \sum_{l=1}^{{A_n}-1} (-1)^{l} \binom{A_n-1}{l} \int_0^\infty e^{-ls} ds \\
        & = ({A_n}-1) \sum_{l'=0}^{{A_n}-2} (-1)^{l'} \binom{A_n-2}{l'} \frac{1}{(l'+1)^2}
    \end{split}
\end{equation*}
where in the second equality it was used that (for $0\leq k \leq n$) $\binom{n}{k} = n/k \binom{n-1}{k-1}$ and it was defined $l'=l-1$. Furthermore, the sum in the last equality is just the binomial transform of $1/(l'+1)^2$ which, by equality (8.39) in \cite{boyadzhiev2018binomial}, is equal to $H_{{A_n}-1}/({A_n}-1)$. \end{proof}

As a consequence of this proposition, we can derive an upper bound for the variance of the total number of steps.

\begin{prop}\label{cor:varStilde}
    Given the above construction, we have that 
    \begin{equation*}
        (A_n-1) \left[ A_n/2 + e^{H_{A_n-1}} - (A_n-1)  \right] \leq \Var(\mathcal{\tilde S}) \leq 2 A^2_n + A_n e^{H_{A_n-1}} - (A_n + e^{ H_{A_n-1}})^2.
    \end{equation*}
\end{prop}

\begin{proof}
    We will start by proving the upper bound. First note that, for $i\geq1$, $N_{\bar T_{A_n}}^{(i)} \leq N_{\infty}^{(i)}$. Moreover, the mean value of the right hand side can be explicitly computed and gives
    \begin{equation}\label{eq:mediasemirecta}
        \begin{split}
            \espe\left( N_{\infty}^{(i)} \right) & = \int_0^\infty \left(1-e^{-t}\right)^{i-1} e^{-t} dt = \sum_{k=0}^{i-1} (-1)^k \binom{i-1}{k} \int_0^\infty e^{-(k+1)t}dt\\
            & = \sum_{k=0}^{i-1} (-1)^k \binom{i-1}{k} \frac{1}{k+1} = \frac{1}{i}
        \end{split}
    \end{equation}
    \noindent
    where in the last equality we used relation (9.25) from \cite{boyadzhiev2018binomial}. Making use of \eqref{eq:mediasemirecta}, we can bound the second moment of $\mathcal{\tilde S}\ndep$ according to
    \begin{equation*}
        \begin{split}
            \espe\left[ \left( \mathcal{\tilde S}\ndep \right)^2 \right] \leq & \ \espe\left[\left(A_n + \sum_{i=1}^{A_n-1} i N_\infty^{(i)}\right)^2\right] \\
            & \ \ \ \ = \sum_{i=1}^{A_n-1} i^2 \Var\left(N_\infty^{(i)}\right) + \sum_{i=1}^{A_n-1} i^2 \espe\left(N_\infty^{(i)}\right)^2 \\
            & \ \ \ \ \ \ \ \ \ \ \ \ \ \ \ \ \ \ + \sum_{j<i}^{A_n-1} i j \espe\left(N_\infty^{(i)} N_\infty^{(j)}\right) + A_n e^{H_{A_n-1}} + A_n^2\\
            & \ \ \ \ = \sum_{i=1}^{A_n-1} i^2 \espe\left(N_\infty^{(i)}\right) + \sum_{i=1}^{A_n-1} i^2 \espe\left(N_\infty^{(i)}\right)^2 \\
            & \ \ \ \ \ \ \ \ \ \ \ \ \ \ \ \ \ \ + \sum_{j<i}^{A_n-1} i j \espe\left(N_\infty^{(i)}\right) \espe\left(N_\infty^{(j)}\right) + A_n e^{H_{A_n-1}} + A_n^2 \\
            & \ \ \ \ = \sum_{i=1}^{A_n-1} i + \sum_{i=1}^{A_n-1} 1 + \sum_{j<i}^{A_n-1} 1 + A_n e^{H_{A_n-1}} + A_n^2 \\
            & \ \ \ \ = 2 A_n^2 + A_n e^{H_{A_n-1}} - 1,
        \end{split}
    \end{equation*}
    \noindent
    where, for the first equality, we used the previous proposition and, for the second one, we used the independence of $N_\infty^{(i)}$ from $N_\infty^{(j)}$ and that, because they are Poisson processes, $\Var(N_\infty^{(i)}) = \espe(N_\infty^{(i)})$. Then, subtracting $\espe^2(\mathcal{\tilde S})$ from both sides of the inequality we obtain the upper bound of the proposition.
    
    For the lower bound, observe that by the \emph{law of total covariance} we have that for every distinct pair $i,j\in[A_n-1]$
    \begin{equation*}
        \begin{split}
            \Cov\left(N^{(i)}_{\bar T_{A_n}},N^{(j)}_{\bar T_{A_n}}\right)) & = \espe\left[\Cov(N^{(i)}_{\bar T_{A_n}},N^{(j)}_{\bar T_{A_n}}|\bar T_{A_n})\right] + \Cov\left[\espe(N^{(i)}_{\bar T_{A_n}}|\bar T_{A_n}),\espe(N^{(j)}_{\bar T_{A_n}}|\bar T_{A_n})\right] \\
            & = \Cov\left[\espe(N^{(i)}_{\bar T_{A_n}}|\bar T_{A_n}),\espe(N^{(j)}_{\bar T_{A_n}}|\bar T_{A_n})\right] \geq 0,
        \end{split}
    \end{equation*}
    where in the second line we used that, conditional on $\bar T_{A_n}$, both variables are independent and that their conditional expectations are increasing functions of $\bar T_{A_n}$ (and therefore have positive covariance). We then get that
    \begin{equation*}
        \begin{split}
            \Var\left(\mathcal{\tilde S}\right) & = \sum_{i=1}^{A_n-1} i^2 \Var(N^{(i)}_{\bar T_{A_n}}) + \sum_{i\neq j}^{A_n-1} ij \ \Cov(N^{(i)}_{\bar T_{A_n}},N^{(j)}_{\bar T_{A_n}}) \\
            & \geq  \sum_{i=1}^{A_n-1} i^2 \Var\left[\espe( N^{(i)}_{\bar T_{A_n}}|\bar T_{A_n})\right] + \sum_{i=1}^{A_n-1} i^2 \espe\left[\Var( N^{(i)}_{\bar T_{A_n}}|\bar T_{A_n})\right] \\
            & \geq \sum_{i=1}^{A_n-1} i^2 \espe( N^{(i)}_{\bar T_{A_n}}) = \sum_{i=1}^{A_n-1} i^2 \espe( N^{(i)}_{\infty}) + \sum_{i=1}^{A_n-1} i^2 [\espe( N^{(i)}_{\bar T_{A_n}} )- \espe( N^{(i)}_{\infty})] \\
            & \geq \sum_{i=1}^{A_n-1} i + (A_n-1) \sum_{i=1}^{A_n-1} i \left[\espe( N^{(i)}_{\bar T_{A_n}} )- \frac{1}{i}\right] \\
            & = \frac{(A_n-1)A_n}{2} + (A_n-1) e^{H_{A_n-1}} - (A_n-1)^2,
        \end{split}
    \end{equation*}
    where in the second line we used the positivity of the covariance, in the third one that $\Var( N^{(i)}_{\bar T_{A_n}}|\bar T_{A_n}) = \espe( N^{(i)}_{\bar T_{A_n}}|\bar T_{A_n})$, and in the last one that $\sum_{i=1}^{A_n-1} i \espe( N^{(i)}_{\bar T_{A_n}} ) = e^{H_{A_n-1}}$.
\end{proof}
In particular, this last proposition implies that $\Var(\mathcal{\tilde S})$ scales as $A^2_n$, which means that $\mathcal{\tilde S}$ does not converge to a constant in the limit. Note that the lower bound bound is asymptotically equal to the expression $(1/2+\gamma - 1)A_n^2 \approx 0.0772 A^2_n$; while the upper bound is asymptotically $(1-\gamma(1+\gamma))A_n^2 \approx 0.0896 A_n^2$.

\section{Proofs of the main results}\label{sec:proofspot}

In this section we present the proofs of the results of Section \ref{sec:mainpot}. Along many of them we make use of the fact that we can couple a realisation of the cost exploration process and the AGDs so that the later is equal to the former restricted to some random set of times. 

As mentioned before, by defining a sequence $\suce{W}{i}$ of independent uniform $[0,1]$ random variables, we couple the sequence of AGDs $\suce{X\ndep}{i}$ so that each time a new vertex is explored, its value is taken sequentially from that sequence. We can then also couple the whole sequence of best AGDs to a single cost exploration: the one associated to $\suce{W}{i}$, which we denote $\suce{\tilde X}{i}$. Under this coupling we have that for every $n,i\geq1$, $X\ndep_i = \tilde X_{F\ndep_i}$ (where $F\ndep_i$ is the number of explored vertices up to step $i$).


\subsection{Proof of Lemma \ref{lem:BRAoptimality}}

Here we will parametrise the time evolution of the algorithms studied in a slightly different way. We will assume that during each step only one vertex is explored. That is, during each step a random unexplored vertex is selected and if it is found to be connected to the current vertex, then its cost value is revealed and it is declared \emph{explored}. In this way, a local search algorithm can (during each step) either move to a vertex in the currently explored portion of the neighbourhood of the current vertex or stay in the same vertex and keep exploring its neighbourhood. As before, we will describe the evolution of the local search algorithms by the associated sequence of states $(X_i,U_i)_{i\geq1}$; where, $X_i$ is the cost of the vertex the algorithm is during step $i$ and $U_i$ is the number of vertices explored up to step $i$. In the case of AGD, this sequence is a Markov chain.

We will first prove the conclusion of the lemma for a variation of the AGD where during each step if it explores a vertex with a lower cost, it automatically jumps to it. Suppose that the evolution of this process is given by the sequence $(X_i,U_i)_{i\geq1}$ and there is some other local search algorithm with evolution $(\bar X_i,\bar U_i)_{i\geq1}$. To prove this lemma we will couple both graphs explored in such a way that, when an edge is found in the graph explored by one process, there is also one in the other. We will also couple the sequence of cost values revealed $(W_i)_{i\geq1}$, so that it is the same for both processes. This means that for all $i\geq1$, the $i$-th value revealed by both processes will be the same and equal to $W_i$.

We want to show that this coupling ensures that the realisations of both processes (until one of them finds a solution) is such that that $X_i \geq \bar X_i$ and $U_i = \bar U_i$. For this, assume that initially both processes have the same values. We will show that if the processes at step $i\geq1$ satisfy these relations, they still satisfy them at step $i+1$. Indeed, because by the coupling described an edge is found in one graph iff an edge is found on the other one, then during step $i+1$ we will have the same number of explored vertices in both processes ($U_{i+1}=\bar U_{i+1}$). Furthermore, because the values of the new costs explored are coupled to be equal for both processes and because the $X_{i+1}$ is always equal to the lowest cost value explored so far, this means that $X_{i+1}\geq\bar X_{i+1}$. This then advances the induction.

Now, define $T_j$ as the time (at step $j\geq1$) since the last change of value of $X_i$, and define $\bar T_j$ in an analogous way for $\bar X_i$. The process $X_i$ has found a solution at step $i\geq1$ iff $T_i = U_{i-T_i}$. Analogously, the process $\bar X_i$ has not found a solution at step $i\geq1$ if either $\bar T_i < \bar U_{i-\bar T_i}$ or if a vertex of lower cost value was found since the last jump. But note that, under this coupling, if $\bar T_i > T_i$ then the vertex $\bar v_i$ associated to the cost $\bar X_i$ is not a solution. Indeed, if $\bar T_i > T_i$, it means that the process $X_i$ has jumped after the last jump of $\bar X_i$. But this implies that some vertex with cost lower than $\bar X_i$ has been found in some of the neighbourhoods of $\bar v_i$, and therefore it is not a solution. Hence $\mathcal{E} \leq \mathcal{\bar E}$ under this coupling.

Finally, to finish the proof of the lemma it is enough to note that the value of the minimum found by AGD is equal to the one found by this variant. For this observe that the difference between the number of vertices explored by the AGD and this processes is at most equal to the number of vertices explored by the AGD in the last step (that is, it is of $\mathcal{O}_\mathbb{P}(1)$). Thus, the difference will be $\mathcal{O}_\mathbb{P}(1)$.


\subsection{Proof of Theorem \ref{thm:sparsemean}}

As noted before, the AGD is distributed as the cost exploration discussed in the previous section but restricted to a random collection of times. This fact will be central throughout the proof.

We now define three events, under whose compliment the stopping time of the algorithm will be close to the one of the cost exploration process. Let $\epsilon > 0$ be a small constant to be fixed later on and define:
\begin{itemize}
    \item $A_\epsilon := \{ \mathcal{V}\ndep \geq n^\epsilon \}$. Then, this is the event under which \emph{too many vertices are visited}.
    \item $B_\epsilon := \left\{ \exists j \in[\mathcal{V}\ndep] \ s.t. \ \big|D\ndep_j n^{\alpha-1} - \lambda\big| > \epsilon \right\}$; where we defined (for $j\geq1$) $D_j^{(n)}$ as the size of the neighbourhood of the $j$-th vertex visited by the algorithm. Then, this is the event under which \emph{there is some neighbourhood visited that is either too large or too small}.
    \item $C_\epsilon := \left\{\bigg| \frac{1}{\mathcal{S}\ndep} \sum_{i=1}^{\mathcal{S}\ndep} E_i\ndep - \beta \bigg| > \epsilon \right\}$; where, as before, (for $i\geq1$) $E_i^{(n)}$ is the number of explored vertices during the $i$-th step of the algorithm. Thus, this is the event under which \emph{the mean number of vertices explored per step is either too large or too small}.
\end{itemize}

By estimating the probabilities of these \emph{bad} events, we can determine the asymptotic behaviour of the total number of steps during the process. But first we will prove that under the compliment of these events, the number of steps $\mathcal{S}\ndep$ can be estimated by the corresponding stopping time of the associated cost exploration process, i.e.,
fixing 
$$A_n = \lambda n^{1-\alpha}.$$

Given the cost exploration process $\tilde X_i$ coupled with the sequence of AGDs $X\ndep_i$, define the following stopping times (for $\epsilon > 0$):
\begin{equation*}
    \mathcal{\tilde S}\ndep:=\inf\{ j\geq1: \tilde X_j = \tilde X_{j-\lceil \lambda n^{1-\alpha}\rceil}\}
\end{equation*}
\begin{equation*}
    \mathcal{\tilde S}\ndep_{\epsilon^+}:=\inf\{ j\geq1: \tilde X_j = \tilde X_{j-\lceil (\lambda + \epsilon) n^{1-\alpha}\rceil}\}
\end{equation*}
\begin{equation*}
    \mathcal{\tilde S}\ndep_{\epsilon^-}:=\inf\{ j\geq1: \tilde X_j = \tilde X_{j-\lfloor (\lambda - \epsilon) n^{1-\alpha}\rfloor}\}
\end{equation*}

\begin{lem}\label{lem:cotaStilde}
    Conditionally on the events $B^c_\epsilon$ and $C^c_\epsilon$, we have $\frac{\mathcal{\tilde S}\ndep_{\epsilon^-}}{\beta+\epsilon}\leq \mathcal{S}\ndep \leq \frac{\mathcal{\tilde S}\ndep_{\epsilon^+}}{\beta-\epsilon}$ and $\mathcal{\tilde S}\ndep_{\epsilon^-}\leq \mathcal{E}\ndep \leq \mathcal{\tilde S}\ndep_{\epsilon^+}$.
\end{lem}
\begin{proof}
    First note that
    \begin{equation*}
        \mathcal{S}\ndep = \frac{\mathcal{E}\ndep}{\frac{1}{\mathcal{S}\ndep}\mathcal{E}\ndep} = \frac{\mathcal{E}\ndep}{\frac{1}{\mathcal{S}\ndep}\sum_{i=1}^{\mathcal{S}\ndep} E\ndep_i}.
    \end{equation*}
    And conditionally on the event $C_\epsilon^c$, $\sum_{i=1}^{\mathcal{S}\ndep}  E\ndep_i \geq (\beta - \epsilon) \mathcal{S}\ndep$. Then
    \begin{equation*}
        \mathcal{S}\ndep \leq \frac{\mathcal{E}\ndep}{\beta - \epsilon}.
    \end{equation*}
    Furthermore, on the event $B_\epsilon^c$ no neighbourhood of a visited vertex is larger than $(\lambda + \epsilon)n^{1-\alpha}$, then by the coupling presented at the beginning of the section we will have that $\mathcal{E}\ndep \leq \mathcal{\tilde S}\ndep_{\epsilon^+}$. Which results in
    \begin{equation*}
        \mathcal{S}\ndep \leq \frac{\mathcal{\tilde S}\ndep_{\epsilon^+}}{\beta - \epsilon}.
    \end{equation*}
    The lower bound is obtained in an analogous way.
\end{proof}

\begin{proof}[Proof of Theorem \ref{thm:sparsemean}]
    Let us first give estimations for the probabilities of $A_\epsilon$, $B_\epsilon$, and $C_\epsilon$. By the fact that each visited vertex has at most $n$ neighbours and the coupling with the cost exploration process, we can apply Corollary \ref{cor:boundV} with $A_n$ equal to $n$ to obtain that
    \begin{equation}\label{eq:cotaA}
        \proba(A_\epsilon) = \mathcal{O}(e^{-n^\epsilon}).
    \end{equation}
    For an estimation of the second event, we assume that the event $A_\epsilon^c$ holds. We then have that
    \begin{equation}\label{eq:cotaB}
        \proba(B_\epsilon|A_\epsilon^c) \leq n^\epsilon \proba(\big|D\ndep_1 n^{\alpha-1} - \lambda\big| > \epsilon) \leq n^\epsilon e^{-(\lambda + \epsilon) n^{1-\alpha} \ln(1+\epsilon/\lambda)} (1+o(1)),
    \end{equation}
    where we used that $D\ndep_1 \sim \mbox{Binom}(\lambda n^{-\alpha},n)$ and Chernoff's inequality. Finally, under events $A^c_\epsilon$ and $B^c_\epsilon$, the number of unexplored vertices during the whole process is going to be $n(1+o(1))$. Then, conditioning on these events, the number of explored vertices on each step is well approximated (in the sense that it can be coupled to match the value up to a $o(1)$ correction) 
    by i.i.d. Poisson variables of mean $\beta$. And because $\mathcal{S}^{(n)}$ is a.s. larger than the size of the smallest neighbourhood visited, which under event $B^c_\epsilon$ is larger or equal to $n^{1-\alpha} (\lambda-\epsilon)$, we have that
    \begin{equation}\label{eq:cotaC}
        \proba(C_\epsilon|A_\epsilon\cap B_\epsilon,\mathcal{S}^{(n)}) \leq \frac{\mathcal{S}^{(n)}\beta}{(\mathcal{S}^{(n)}\epsilon)^2} = \mathcal{O}(n^{-(1-\alpha)}).
    \end{equation}
    
    Defining the event $D_\epsilon = A^c_\epsilon \cap B^c_\epsilon \cap C^c_\epsilon$, the difference in $L^1$ between $\mathcal{S}\ndep/n^{1-\alpha}$ and $\mathcal{\tilde S}\ndep/\beta n^{1-\alpha}$ is given by
    \begin{equation}\label{eq:partoespe}
        \begin{split}
            \espe\left(\bigg|\frac{\mathcal{S}\ndep}{n^{1-\alpha}}-\frac{\mathcal{\tilde S}\ndep}{\beta n^{1-\alpha}}\bigg|\right) & = \espe\left(\bigg|\frac{\mathcal{S}\ndep}{n^{1-\alpha}}-\frac{\mathcal{\tilde S}\ndep}{\beta n^{1-\alpha}}\bigg|\mathbb{I}_{D_\epsilon}\right) + \espe\left(\frac{\mathcal{S}\ndep}{n^{1-\alpha}}\mathbb{I}_{A_\epsilon}\right) \\
            & \ \ \ \ \ \ \ \ \ \ \ \ + \espe\left(\frac{\mathcal{S}\ndep}{n^{1-\alpha}}\mathbb{I}_{A_\epsilon^c \cap B_\epsilon}\right) \espe\left(\frac{\mathcal{S}\ndep}{n^{1-\alpha}}\mathbb{I}_{A_\epsilon^c\cap B_\epsilon^c \cap C_\epsilon}\right) + \espe\left(\frac{\mathcal{\tilde S}\ndep}{\beta n^{1-\alpha}}\mathbb{I}_{D^c_\epsilon}\right),
        \end{split}
    \end{equation}
    where we used that $D^c_\epsilon = A_\epsilon \cup (A_\epsilon^c \cap B_\epsilon) \cup (A_\epsilon^c\cap B_\epsilon^c\cap C_\epsilon)$. We will now bound the different terms in the right hand side of this equation. For the first one note that
    \begin{equation*}
            \espe\left(\bigg|\frac{\mathcal{S}\ndep}{n^{1-\alpha}}-\frac{\mathcal{\tilde S}\ndep}{\beta n^{1-\alpha}}\bigg|\mathbb{I}_{D_\epsilon}\right) \leq \espe\left(\bigg|\frac{\mathcal{\tilde S}_{\epsilon^+}\ndep}{(\beta-\epsilon)n^{1-\alpha}}-\frac{\mathcal{\tilde S}_{\epsilon^-}\ndep}{(\beta+\epsilon) n^{1-\alpha}}\bigg|\right)
    \end{equation*}
    as $\mathcal{\tilde S}\ndep_{\epsilon^-} \leq \mathcal{\tilde S}\ndep \leq \mathcal{\tilde S}\ndep_{\epsilon^+}$ and by Lemma \ref{lem:cotaStilde} under $D_\epsilon$ we have that $\frac{\mathcal{\tilde S}\ndep_{\epsilon^-}}{\beta+\epsilon} \leq \mathcal{S}\ndep \leq \frac{\mathcal{\tilde S}\ndep_{\epsilon^+}}{\beta-\epsilon}$. Furthermore, because $\mathcal{\tilde S}\ndep_{\epsilon^-} \leq \mathcal{\tilde S}\ndep_{\epsilon^+}$ then the absolute value can be removed from the r.h.s. of the inequality. Then, by Proposition \ref{prop:meanconvergence} we have that
    \begin{equation}\label{eq:cotaterm1}
        \espe\left(\bigg|\frac{\mathcal{S}\ndep}{n^{1-\alpha}}-\frac{\mathcal{\tilde S}\ndep}{\beta n^{1-\alpha}}\bigg|\mathbb{I}_{D_\epsilon}\right) \leq \lambda e^\gamma \left( \frac{1}{\beta - \epsilon} - \frac{1}{\beta + \epsilon} \right) (1+o(1)) = \mathcal{O}(\epsilon).
    \end{equation}
    Choosing $\epsilon$ sufficiently small, the next three terms in \eqref{eq:partoespe} can be bounded using \eqref{eq:cotaA}-\eqref{eq:cotaC}:
    \begin{itemize}
        \item $\espe\left(\frac{\mathcal{S}\ndep}{n^{1-\alpha}}\mathbb{I}_{A_\epsilon}\right) \leq n^\alpha \proba(A_\epsilon) = o(1)$,
        \item $\espe\left(\frac{\mathcal{S}\ndep}{n^{1-\alpha}}\mathbb{I}_{A_\epsilon^c \cap B_\epsilon}\right) \leq n^{\epsilon+\alpha} \proba(B_\epsilon|A_\epsilon^c) = o(1)$, and
        \item $\espe\left(\frac{\mathcal{S}\ndep}{n^{1-\alpha}}\mathbb{I}_{A_\epsilon^c\cap B_\epsilon^c \cap C_\epsilon}\right) \leq (\lambda+\epsilon) n^\epsilon \proba(C_\epsilon|A_\epsilon^c\cap B_\epsilon^c) = o(1)$.
    \end{itemize}
    Finally, for the last term observe that by the bounds \eqref{eq:cotaA}-\eqref{eq:cotaC} $\proba(D^c_\epsilon)=o(1)$. And by Proposition \ref{prop:meanconvergence} and Corollary \ref{cor:varStilde},
    \begin{equation*}
        \espe\left(\left(\frac{\mathcal{\tilde S}\ndep}{n^{1-\alpha}}\right)^2\right) = \mathcal{O}(1).
    \end{equation*}
     Therefore, by Cauchy-Schwarz inequality we get that
     $\espe\left(\mathcal{\tilde S}\ndep/(\beta n^{1-\alpha})\mathbb{I}_{D^c_\epsilon}\right) = o(1)$. We then conclude that
    \begin{equation*}
        \espe\left(\bigg|\frac{\mathcal{S}\ndep}{n^{1-\alpha}}-\frac{\mathcal{\tilde S}\ndep}{\beta n^{1-\alpha}}\bigg|\right) = \mathcal{O}(\epsilon) + o(1).
    \end{equation*}
    Because $\epsilon$ is arbitrarily small, we get that $\espe\left(\bigg|\frac{\mathcal{S}\ndep}{n^{1-\alpha}}-\frac{\mathcal{\tilde S}\ndep}{\beta n^{1-\alpha}}\bigg|\right) \xrightarrow{n\rightarrow\infty}0$.
    
    The proof that
    \begin{equation*}
        \espe\left(\big|\mathcal{E}\ndep-\mathcal{\tilde S}\ndep\big|\right) = \mathcal{O}(\epsilon) + o(1)
    \end{equation*}
    follows in a similar way.
\end{proof}

\bibliographystyle{alpha}
\bibliography{biblio}

\end{document}